\documentclass[onecolumn,12pt]{IEEEtran} 
\usepackage{color}
\usepackage{amsfonts,color,pslatex}
\usepackage{amssymb,amsmath,latexsym}

\usepackage[para]{threeparttable}

\newcommand{\tr}{{\rm Tr}}
\newcommand{\gf}{{\mathbb{F}}}

\newcommand{\C}{{\mathcal C}}

\newtheorem{theorem}{Theorem}[section]
\newtheorem{lemma}[theorem]{Lemma}

\newtheorem{definition}[theorem]{Definition}
\newtheorem{example}[theorem]{Example}

\begin{document}

\title{A Family of Five-Weight Cyclic Codes and Their Weight Enumerators\thanks{Z. Zhou's research was supported by
the Natural Science Foundation of China, Proj. No. 61201243. C. Ding's research was supported by
The Hong Kong Research Grants Council, Proj. No. 600812.
J. Luo was supported the Norwegian Research Council
under Grant No. 191104/V30, the National Science
Foundation (NSF) of China under Grant No. 60903036, the NSF of
Jiangsu Province under Grant No. 2009182, and the Open Research Fund
of the National Mobile Communications Research Laboratory, Southeast University
(No. 2010D12).}
}
\author{Zhengchun Zhou,\thanks{Z. Zhou is with the School of Mathematics, Southwest Jiaotong University,
Chengdu, 610031, China (email: zzc@home.swjtu.edu.cn).}
Cunsheng Ding,\thanks{C. Ding is with the Department of Computer Science
                                                  and Engineering, The Hong Kong University of Science and Technology,
                                                  Clear Water Bay, Kowloon, Hong Kong, China (email: cding@ust.hk).}
Jinquan Luo\thanks{J. Luo is with the Department of Informatics, University of Bergen,
N-5020 Bergen, Norway (email: Jinquan.Luo@ii.uib.no).}
and
             Aixian Zhang\thanks{A. Zhang is with the School of Mathematical Sciences, Capital Normal University,
Beijing 100048, P. R. China (email: zhangaixian1008@126.com).}
}
\date{\today}
\maketitle

\begin{abstract}
Cyclic codes are a subclass of linear codes and have applications in consumer electronics, data storage
systems, and communication systems as they have efficient encoding and decoding algorithms. In this
paper, a family of $p$-ary cyclic codes whose duals have three zeros are proposed. The weight distribution
of this family of cyclic codes is  determined. It turns out that the proposed cyclic codes have five nonzero
weights.
\end{abstract}

\section{Introduction}\label{sec-intro}
An $[n,\ell,d]$ linear code over the finite field $\gf_p$ is an $\ell$-dimensional subspace of $\gf_p^n$
with minimum (Hamming) distance $d$, where $p$ is a prime. Let $A_i$ denote the number of codewords
with Hamming weight $i$ in a code $\mathcal{C}$ of length $n$. The weight enumerator of $\mathcal{C}$
is defined by
\begin{eqnarray*}
1+A_1z+A_2z^2+\cdots+A_nz^n.
\end{eqnarray*}
The sequence $(A_1,A_2,\cdots,A_{n})$ is called the weight distribution of the code. Clearly, the weight
distribution gives the minimum distance of the code, and thus the error correcting capability. In addition, the
weight distribution of a code allows the computation of the error probability of error detection and correction
with respect to some error detection and error correction algorithms \cite{Klov}. Thus the study of the weight
distribution of a linear code is important in both theory and applications.

An  $[n,k]$ linear code $\C$ over  $\gf_p$ is called {cyclic} if
$(c_0,c_1, \cdots, c_{n-1}) \in \C$ implies $(c_{n-1}, c_0, c_1, \cdots, c_{n-2})$
$\in \C$.
By identifying any vector $(c_0,c_1, \cdots, c_{n-1}) \in \gf_p^n$
with
$$
c_0+c_1x+c_2x^2+ \cdots + c_{n-1}x^{n-1} \in \gf_p[x]/(x^n-1),
$$
any code $\C$ of length $n$ over $\gf_p$ corresponds to a subset of $\gf_p[x]/(x^n-1)$.
The linear code $\C$ is cyclic if and only if the corresponding subset in $\gf_p[x]/(x^n-1)$
is an ideal.
It is well known that every ideal of $\gf_p[x]/(x^n-1)$ is principal. Let $\C=\langle g(x) \rangle$,
where $g(x)$ is monic and has the least
degree. Then $g(x)$ is called the  generator polynomial and
$h(x)=(x^n-1)/g(x)$ is referred to as the  parity-check polynomial of
$\C$.  A cyclic code is called irreducible
if its parity-check polynomial is irreducible over $\gf_p$. Otherwise, it is called reducible.

The weight distributions of both irreducible and reducible cyclic codes have been interesting subjects
of study for many years and are very hard problems in general. For information on the weight distribution
of irreducible cyclic codes, the reader
is referred to the recent survey \cite{Ding12}. Information
on the weight distribution of reducible cyclic codes could be found in \cite{YCD}, \cite{Feng07}, \cite{Luo081},
\cite{Luo082}, \cite{Zeng10}, \cite{Ma11}, \cite{Ding11}, \cite{Wang12}.

For the duals of the known cyclic codes whose weight distributions were established, most of them have at most two zeros  (see \cite{YCD,Feng07,Luo081,Luo082,Ma11,Ding11,Wang12,Xiong1,Feng12,Ding12}),
only a few of them have three or more zeros  (see \cite{Feng07,Luo081,Zeng10,Li12}).
The objective of this paper is to settle the weight distribution of  a family of five-weight cyclic codes
whose duals have three zeros.

This paper is organized as follows. Section \ref{sec-code} defines the  family of cyclic codes.
Section \ref{sec-prelim} presents results on quadratic forms  which will be needed in subsequent
sections. Section \ref{sec-wtds}  solves the
weight distribution problem for the family of cyclic codes. Section \ref{sec-conclusion} concludes
this paper and makes some comments.

\section{The  family of cyclic codes}\label{sec-code}

In this section, we introduce the family of cyclic codes to be studied in the sequel.
Before doing this, we first give some notations which will be fixed throughout the paper unless otherwise stated.

Let $p$ be an odd prime and $q=p^m$, where $m$ is odd and $m\geq 5$.
Let $d_1=(p^{2k}+1)/2$ and $d_2=(p^{4k}+1)/2$, where $k$ is any positive integer with $\gcd(m,k)=1$.
Let $\pi$ be a generator of the finite field $\gf_{q}$, and let $h_i(x)$ denote the minimal polynomial of $\pi^{-i}$ over $\gf_p$
for any integer $i$. It is easy to check that  $h_1(x)$, $h_{d_1}(x)$ and $h_{d_2}(x)$ have
degree $m$ and are pairwise distinct. Define
\begin{eqnarray}\label{eqn-parity-check}
h(x)=h_1(x)h_{d_1}(x)h_{d_2}(x).
\end{eqnarray}
Then $h(x)$ has degree $3m$ and is a factor of $x^{q-1}-1$.

Let $\C_{(p,m,k)}$ be the cyclic code with parity-check polynomial $h(x)$. Then
$\C_{(p,m,k)}$ has length $q-1$ and dimension $3m$. Using the well-known Delsarte's
Theorem \cite{Delsarte}, one can prove that
\begin{eqnarray}\label{eqn-def-code-C}
\C_{(p,m,k)}=\{{\bf{c}}_{\Delta}:  \Delta=(\delta_0,\delta_1,\delta_2) \in \gf^3_{q}\}
\end{eqnarray}
where the codeword
\begin{eqnarray*}
{\bf{c}}_{\Delta}=\left(\tr(\delta_0 \pi^i+\delta_1 \pi^{id_1}+\delta_2 \pi^{id_2})\right)_{i=0}^{q-2}
\end{eqnarray*}
and $\tr$ denotes the absolute trace from $\gf_q$ to $\gf_p$.

Let $h'(x)=h_1(x)h_{d_1}(x)$ and $\C_{(p,m,k)}'$ be the cyclic code with parity-check polynomial $h'(x)$.
Then $\C_{(p,m,k)}'$ is a subcode of $\C_{(p,m,k)}$ with dimension $2m$. Trachtenberg \cite{Trachtenberg} proved
that $\C_{(p,m,k)}'$ has three nonzero weights and determined its weight distribution.
The objectives of this paper are to show that $\C_{(p,m,k)}$ have five nonzero weights
and settle the weight distribution of this class of cyclic codes $\C_{(p,m,k)}$.

\section{Mathematical foundations}\label{sec-prelim}

In this section, we give a brief introduction to quadratic forms
over finite  fields which will be useful in the sequel. Quadratic forms  have
been well studied (see the monograph \cite{Niddle} and the references therein), and have applications
in sequence design (\cite{Trachtenberg}, \cite{Klapperodd}, \cite{Tang}),  and coding theory (\cite{Feng07},
\cite{Luo081}, \cite{Luo082}, \cite{Zeng10}).

\begin{definition}
Let $x=\sum_{i=1}^mx_i\alpha_i$ where $x_i\in \gf_p$ and
$\{\alpha_1,\cdots,\alpha_m\}$ is a basis for $\gf_{q}$ over $\gf_p$.
Then a function  $Q(x)$ from $\gf_{q}$ to $\gf_p$ is a quadratic
form over $\gf_{p}$ if it can be represented as
\begin{eqnarray*}
Q(x)=Q\left(\sum_{i=1}^m x_i\alpha_i\right)=\sum_{i=1}^m\sum_{j=1}^mb_{i,j}x_ix_j
\end{eqnarray*}
where $b_{i,j}\in \gf_p$. That is, $Q(x)$ is  a homogeneous
polynomial of degree 2 in the ring $\gf_p[x_1,x_2,\cdots,x_m]$.
\end{definition}

The rank of the
quadratic form $Q(x)$ is defined as the codimension of the $\gf_p$-vector space
\begin{eqnarray*}
V=\{z\in \gf_{q}: Q(x+z)-Q(x)-Q(z)=0 \textrm{~for~all~}x\in \gf_{q}\}.
\end{eqnarray*}
That is $|V|=p^{m-r}$ where $r$ is the rank of $Q(x)$.

In order to determine  the weight distribution of the aforementioned code $\C_{(p,m,k)}$,
we need to deal with the
exponential sum of the following form:
\begin{eqnarray}\label{eqn_s_f}
S_{f}=\sum_{x\in \gf_{q}}\zeta_p^{f(x)}
\end{eqnarray}
where $\zeta_p$ is a complex primitive $p$-th root of unity, and $f(x)$ is a
function from $\gf_q$ to $\gf_p$ satisfying
\begin{enumerate}
\item $f(yx)=yf(x)$ for all $y\in \gf_p$;

\item $Q(x)=f(x^2)$ is a quadratic form over $\gf_p$.

\end{enumerate}

Note that any nonsquare  in $\gf_p$
is also a nonsquare in $\gf_q$ since $m$ is odd. It is easy to verify that
\begin{eqnarray}\label{eqn_oneform_twoform}
2S_f=\sum_{x\in \gf_q}\zeta_p^{Q(x)}+\sum_{x\in \gf_q}\zeta_p^{\lambda Q(x)}
\end{eqnarray}
where $\lambda$ is a fixed nonsquare in $\gf_p$.
The following result can be traced back to Trachtenberg \cite{Trachtenberg}
whose proof is based on (\ref{eqn_oneform_twoform}) and the classification
of quadratic forms over finite fields in odd characteristic.
For more details, we refer the reader to Pages 30-36 in
\cite{Trachtenberg} and Lemma 4 in \cite{Tang}.

\begin{lemma}\label{Lemma-berg}
Let $S_f$ be defined
by (\ref{eqn_s_f}) and $r$ be the rank of the quadratic form $Q(x)=f(x^2)$.
Then $S_f=0$ if $r$ is odd, and $S_f=\pm p^{m-{r/2}}$ otherwise.
\end{lemma}

\section{The Weight distribution of the family of cyclic codes}\label{sec-wtds}

In this section, we shall establish the weight distribution
of the code $\C_{(p,m,k)}$ of (\ref{eqn-def-code-C}) defined in Section \ref{sec-code}.
To this end, we need a series of lemmas. Before introducing them, for any  $\Delta=(\delta_0,\delta_1,\delta_2)\in \gf_{q}^3$, we define
\begin{eqnarray}\label{eqn-f-Delta}
f_{\Delta}(x)=\tr(\delta_0 x+\delta_1 x^{d_1}+\delta_2 x^{d_2}), ~ x\in \gf_q
\end{eqnarray}
and
\begin{eqnarray}\label{eqn-S-f-delta}
S_{f_{\Delta}}=\sum_{x\in \gf_q}\zeta_p^{f_{\Delta}(x)}.
\end{eqnarray}

\begin{lemma}\label{lemma-rankmain}
Let $f_{\Delta}(x)$ be defined by (\ref{eqn-f-Delta}).
Then $f_{\Delta}(yx)=yf_{\Delta}(x)$ for any $y\in \gf_p$.
And for any $\Delta\neq (0,0,0)$,  the quadratic form
$Q_{\Delta}(x)=f_{\Delta}(x^2)$ has rank $m-i$ for some $0\leq i\leq 4$.
\end{lemma}

\begin{proof}
Recall that $d_1=(p^{2k}+1)/2$ and $d_2=(p^{4k}+1)/2$.
Thus $y^{d_1}=y$ and $y^{d_2}=y$ for any $y\in \gf_p$.
This together with the linear properties of the trace function
means that $f_{\Delta}(yx)=yf_{\Delta}(x)$ for any $y\in \gf_p$.
Clearly,  $Q_{\Delta}(x)=f_{\Delta}(x^2)=\tr(\delta_0 x^2+\delta_1 x^{p^{2k}+1}+\delta_2 x^{p^{4k}+1})$
is a quadratic form over $\gf_p$.
We now calculate the rank of $Q_{\Delta}(x)$.
Note that
\begin{eqnarray*}
Q_{\Delta}(x+z)-Q_{\Delta}(x)-Q_{\Delta}(z)=\tr(z L_{\Delta}(x))
\end{eqnarray*}
where
\begin{eqnarray*}
L_{\Delta}(x)=2\delta_0 x+\delta_1 x^{p^{2k}}+\delta_1^{p^{-2k}}x^{p^{-2k}}+\delta_2 x^{p^{4k}}+\delta_2^{p^{-4k}} x^{p^{-4k}}.
\end{eqnarray*}
We need to calculate the number of roots of the linearized polynomial $L_{\Delta}(x)$. Let $H_{\Delta}(x)=(L_{\Delta}(x))^{p^{4k}}$. Then
\begin{eqnarray}\label{eqn_g(z)}
H_{\Delta}(x)=\delta_2^{p^{4k}} x^{p^{8k}}+\delta_1^{p^{4k}} x^{p^{6k}}+\delta_1^{p^{2k}}x^{p^{2k}}+2\delta_0^{p^{4k}}x^{p^{4k}}+\delta_2 x.
\end{eqnarray}
Clearly,  $L_{\Delta}(x)$ has the same number of roots in $\gf_{p^m}$ as
$H_{\Delta}(x)$.
Fix an algebraic closure $\gf_{p^\infty}$ of $\gf_p$, then all
roots of $H_{\Delta}(x)$ form a vector space over $\gf_{p^{2k}}$ of dimension at most
$4$ since its degree is at most $p^{8k}=(p^{2k})^4$ for any $(\delta_0,\delta_1,\delta_2)\neq (0,0,0)$.
Note that $\gcd(m,2k)=1$, it is straightforward (see Lemma 4,
\cite{Trachtenberg}) to verify that elements in $\gf_{p^m}$
that are linearly independent over $\gf_{p}$ are also linearly
independent over $\gf_{p^{2k}}$. Therefore, the roots of $H_{\Delta}(x)$
 in $\gf_{p^m}$ form a vector space over $\gf_{p}$ of
dimension at most $4$. Thus the rank of $Q_{\Delta}(x)$ is at
least $m-4$ for any $\Delta\neq (0,0,0)$. This completes the proof.
\end{proof}
\vspace{2mm}

\begin{lemma}\label{lem-N2}
Let $\mathfrak{N}_2$ denote the number of solutions $(x,y) \in \gf_q^2$
of the following system of equations
\begin{eqnarray}\label{eqn-N2}
\left\{
\begin{array}{l}
x+y=0 \\
x^{d_1}+y^{d_1}=0 \\
x^{d_2}+y^{d_2}=0.
\end{array}
\right.
\end{eqnarray}
Then $\mathfrak{N}_2=q.$
\end{lemma}
\begin{proof}
The conclusion follows directly from the observation that  $(x,y)$
is a solution of (\ref{eqn-N2}) if and only if $y=-x$.
\end{proof}

\begin{lemma}\label{lem-N3}
Let $\mathfrak{N}_3$ denote the number of solutions $(x,y,u) \in \gf_q^3$
of the following system of equations
\begin{eqnarray}\label{eqn-N3}
\left\{
\begin{array}{l}
x+y+u=0 \\
x^{d_1}+y^{d_1}+u^{d_1}=0 \\
x^{d_2}+y^{d_2}+u^{d_2}=0.
\end{array}
\right.
\end{eqnarray}
Then $\mathfrak{N}_3=qp+q-p.$
\end{lemma}
\begin{proof}
We distinguish between the following two cases to calculate the number of solutions $(x,y,u) \in \gf_q^3$
of (\ref{eqn-N3}).

{\textit{Case A}}, when $u=0$: In this case,  by Lemma \ref{lem-N2},
the number of solutions of (\ref{eqn-N3}) is equal to $q$.

{\textit{Case B}}, when $u\neq 0$: In this case, for each  $u\in \gf^*_q$,
the equation system (\ref{eqn-N3}) has the same number of solutions $(x,y)\in \gf_q^2$ of
\begin{eqnarray*}
\left\{
\begin{array}{l}
1+x+y=0 \\
1+x^{d_1}+y^{d_1}=0 \\
1+x^{d_2}+y^{d_2}=0
\end{array}
\right.
\end{eqnarray*}
which has the same number of solutions $x\in \gf_q$ of
\begin{eqnarray}\label{eqn-N3-2}
\left\{
\begin{array}{l}
1+x^{d_1}=(1+x)^{d_1} \\
1+x^{d_2}=(1+x)^{d_2}.
\end{array}
\right.
\end{eqnarray}
By performing square on both sides of each equation in (\ref{eqn-N3-2}), we have
\begin{eqnarray*}
\left\{
\begin{array}{l}
x(x^{(p^{2k}-1)/2}-1)=0 \\
x(x^{(p^{4k}-1)/2}-1)=0 \\
\end{array}
\right.
\end{eqnarray*}
which implies that $x\in \gf_p$ since $\gcd(m,2k)=\gcd(m,4k)=1$. Conversely,
for any $x\in \gf_p$, it is clear that $x$ is a solution to (\ref{eqn-N3-2})
since $x^{d_i}=x$ and $(1+x)^{d_i}=1+x$ for each $i=1,2$. Thus (\ref{eqn-N3-2})
has exactly $p$ solutions.

Summarizing the results of the two cases above,  we have that $\mathfrak{N}_3=q+(q-1)p=qp+q-p$. This completes the proof.
\end{proof}

\vspace{2mm}

The following lemma is
the key to establishing the weight distribution
of the proposed code $\C_{(p,m,k)}$. Its proof is lengthy and is
presented in Appendix I.

\vspace{2mm}

\begin{lemma}\label{lem-mainA}
Let $\mathfrak{N}_4$ denote the number of solutions $(x,y,u,v) \in \gf_q^4$
of the following system of equations
\begin{eqnarray}\label{eqn-mainlemmaA}
\left\{
\begin{array}{l}
x+y+u+v=0 \\
x^{d_1}+y^{d_1}+u^{d_1}+v^{d_1}=0 \\
x^{d_2}+y^{d_2}+u^{d_2}+v^{d_2}=0.
\end{array}
\right.
\end{eqnarray}
Then $\mathfrak{N}_4=q(qp+q-p).$
\end{lemma}
\begin{proof}
See Appendix I.
\end{proof}

\vspace{2mm}

\begin{theorem}\label{theorem-ds-S}
Let $S_{f_{\Delta}}$ be defined by (\ref{eqn-S-f-delta}). Then, as $\Delta$ runs
through $\gf_q^3$, the value distribution of $S_{f_{\Delta}}$ is given by Table
\ref{Tb_dis-S}.
\end{theorem}

\begin{table*}[!t]
\renewcommand{\arraystretch}{2}
\centering
\begin{threeparttable}
\caption{Value Distribution of $S_{f_{\Delta}}$}\label{Tb_dis-S}
\begin{tabular}{|l|l|}
\hline
Value& Frequency\\
\hline
\hline
  $p^m$ & $1$ \\
  \hline
 $0$ & $(p^m-1)(p^{2m}-p^{2m-1}+p^{2m-4}+p^m-p^{m-1}-p^{m-3}+1)$\\
\hline
 $p^{(m+1)/2}$ & $\frac{(p^{m+1}+p^{(m+3)/2})(p^{2m}-p^{2m-2}-p^{2m-3}+p^{m-2}+p^{m-3}-1)}{2(p^2-1)}$\\
\hline
$-p^{(m+1)/2}$ & $\frac{(p^{m+1}-p^{(m+3)/2})(p^{2m}-p^{2m-2}-p^{2m-3}+p^{m-2}+p^{m-3}-1)}{2(p^2-1)}$\\
\hline
 $p^{(m+3)/2}$ & $\frac{(p^{m-3}+p^{(m-3)/2})(p^{m-1}-1)(p^m-1)}{2(p^2-1)}$\\
\hline
$-p^{(m+3)/2}$ & $\frac{(p^{m-3}-p^{(m-3)/2}(p^{m-1}-1)(p^m-1)}{2(p^2-1)}$\\
\hline
\end{tabular}
\begin{tablenotes}
\end{tablenotes}
\end{threeparttable}
\end{table*}

\begin{proof}
It is clear that $S_{f_{\Delta}}=p^m$ if  $\Delta=(0,0,0)$.
Otherwise, by Lemmas \ref{lemma-rankmain} and \ref{Lemma-berg}, we have
\begin{eqnarray*}
S_{f_{\Delta}}\in \{0,\pm p^{(m+1)/2}, \pm p^{(m+3)/2}\}.
\end{eqnarray*}
To determine the distribution of these values,  we define
\begin{eqnarray*}
n_{1,i}&=&\# \{\Delta \in \gf_q^3\setminus \{(0,0,0)\}: ~~S_{f_{\Delta}}=(-1)^i p^{(m+1)/2}\},\\
n_{2,i}&=&\# \{\Delta \in \gf_q^3\setminus \{(0,0,0)\}: ~~S_{f_{\Delta}}=(-1)^i p^{(m+3)/2}\}
\end{eqnarray*}
where $i=0,1$. Then the value distribution of $S_{f_{\Delta}}$ is as follows
\begin{eqnarray}\label{eqn-des-S}
\begin{array}{lcl}
p^m&\textrm{occurring}&1~{\textrm{time}}\\
p^{(m+1)/2}&\textrm{occurring}&n_{1,0}~{\textrm{times}}\\
-p^{(m+1)/2}&\textrm{occurring}&n_{1,1}~{\textrm{times}}\\
p^{(m+3)/2}&\textrm{occurring}&n_{2,0}~{\textrm{times}}\\
-p^{(m+3)/2}&\textrm{occurring}&n_{2,1}~{\textrm{times}}\\
0&\textrm{occurring}&p^{3m}-1-n_1-n_2~{\textrm{times}}.
\end{array}
\end{eqnarray}
By (\ref{eqn-des-S}), we immediately have
\begin{eqnarray}\label{eqn-total-S-1}
\left\{
\begin{array}{lcl}
\sum_{\Delta\in \gf^3_q}S_{f_{\Delta}}&=&p^{m}+(n_{1,0}-n_{1,1})p^{(m+1)/2}+(n_{2,0}-n_{2,1})p^{(m+3)/2}\\
\sum_{\Delta\in \gf^3_q}S^2_{f_{\Delta}}&=&p^{2m}+(n_{1,0}+n_{1,1})p^{m+1}+(n_{2,0}+n_{2,1})p^{m+3}\\
\sum_{\Delta\in \gf^3_q}S^3_{f_{\Delta}}&=&p^{3m}+(n_{1,0}-n_{1,1})p^{(3m+3)/2}+(n_{2,0}-n_{2,1})p^{(3m+9)/2}\\
\sum_{\Delta\in \gf^3_q}S^4_{f_{\Delta}}&=&p^{4m}+(n_{1,0}+n_{1,1})p^{2m+2}+(n_{2,0}+n_{2,1})p^{2m+6}.
\end{array}
\right.
\end{eqnarray}
On the other hand, applying Lemmas \ref{lem-N2}, \ref{lem-N3} and \ref{lem-mainA}, we have
\begin{eqnarray}\label{eqn-total-S-2}
\begin{array}{lcl}
\sum_{\Delta\in \gf^3_q}S_{f_{\Delta}}&=&p^{3m}\\
\sum_{\Delta\in \gf^3_q}S^2_{f_{\Delta}}&=&p^{4m}\\
\sum_{\Delta\in \gf^3_q}S^3_{f_{\Delta}}&=&p^{3m}(p^{m+1}+p^m-p)\\
\sum_{\Delta\in \gf^3_q}S^4_{f_{\Delta}}&=&p^{4m}(p^{m+1}+p^m-p).
\end{array}
\end{eqnarray}
Combining Equations (\ref{eqn-total-S-1}) and  (\ref{eqn-total-S-2}) gives
$$
n_{1,0}= \frac{(p^{m+1}+p^{(m+3)/2})(p^{2m}-p^{2m-2}-p^{2m-3}+p^{m-2}+p^{m-3}-1)}{2(p^2-1)},
$$
$$
n_{1,1}= \frac{(p^{m+1}-p^{(m+3)/2})(p^{2m}-p^{2m-2}-p^{2m-3}+p^{m-2}+p^{m-3}-1)}{2(p^2-1)},
$$
$$
n_{2,0}= \frac{(p^{m-3}+p^{(m-3)/2})(p^{m-1}-1)(p^m-1)}{2(p^2-1)},
$$
$$
n_{2,1}= \frac{(p^{m-3}-p^{(m-3)/2})(p^{m-1}-1)(p^m-1)}{2(p^2-1)}.
$$
The value distribution of $S_{f_{\Delta}}$ depicted in Table \ref{Tb_dis-S} then follows
from the values of $n_{1,0}, n_{1,1}, n_{2,0}$ and $n_{2,1}$, and the analysis above.
\end{proof}

The following is the main result of the paper.

\begin{theorem}\label{Th_main1}
Let $\C_{(p,m,k)}$ be the code  in (\ref{eqn-def-code-C}). Then $\C_{(p,m,k)}$ is a
cyclic code over $\gf_p$  with parameters
$$
[p^m-1,3m,(p-1)(p^{m-1}-p^{(m+1)/2})].
$$
Furthermore,  the weight distribution of  $\C_{(p,m,k)}$ is given by Table \ref{Tb_wd}.
\end{theorem}

\begin{proof}
The length and dimension of the code follow directly from the definition of $\C_{(p,m,k)}$. We only need
to determine its minimal weight and weight distribution.
In terms of exponential sums, the weight of the codeword ${\bf{c}}_{\Delta}$ in $\C_{(p,m,k)}$ is given by
\begin{eqnarray}\label{eqn-weight}
{\textrm{WT}}({\bf c}_{\Delta})&=&\#\{x\in \gf^*_q: ~~\tr(\delta_0 x+ \delta_1 x^{d_1}+ \delta_2 x^{d_2})\neq 0\}   \nonumber\\
&=&q-1-\#\{x\in \gf^*_q: ~~\tr(\delta_0 x+ \delta_1 x^{d_1}+ \delta_2 x^{d_2})= 0\}   \nonumber\\
&=&q-1-{1\over p}\sum_{x\in \gf_q^*}\sum_{y\in \gf_p}\zeta_p^{y\tr(\delta_0 x+ \delta_1 x^{d_1}+ \delta_2 x^{d_2})}  \nonumber\\
&=&p^{m}-p^{m-1}-{1\over p}\sum_{y\in \gf^*_p}\sum_{x\in \gf_q}\zeta_p^{\tr(\delta_0 yx+ \delta_1 yx^{d_1}+ \delta_2 yx^{d_2})}\nonumber \\
&=&(p-1)p^{m-1}-{1\over p}\sum_{y\in \gf^*_p}\sum_{x\in \gf_q}\zeta_p^{\tr(\delta_0 yx+ \delta_1 (yx)^{d_1}+ \delta_2 (yx)^{d_2})}  \nonumber\\
&=&(p-1)p^{m-1}-{1\over p}\sum_{y\in \gf^*_p}\sum_{x\in \gf_q}\zeta_p^{\tr(\delta_0 x+ \delta_1 x^{d_1}+ \delta_2 x^{d_2})}\nonumber \\
&=&(p-1)p^{m-1}-{p-1\over p}S_{f_{\Delta}}
\end{eqnarray}
where $S_{f_{\Delta}}$ is given by (\ref{eqn-S-f-delta}) and in the fifth identity we used the fact that
$y^{d_i}=y$ for any $y\in \gf_p$. The minimal weight and weight distribution of $\C_{(p,m,k)}$ then follow from (\ref{eqn-weight})
and the value distribution of the exponential sum $S_{f_{\Delta}}$ depicted in Table \ref{Tb_dis-S}.
\end{proof}

\vspace{2mm}

\begin{table*}[!t]
\renewcommand{\arraystretch}{2}
\centering
\begin{threeparttable}
\caption{Weight Distribution of the Code $\C_{(p,m,k)}$ in Theorem \ref{Th_main1}}\label{Tb_wd}
\begin{tabular}{|l|l|}
\hline
 Hamming Weight& Frequency\\
\hline
\hline
  $0$ & $1$ \\
  \hline
 $(p-1)p^{m-1}$ & $(p^m-1)(p^{2m}-p^{2m-1}+p^{2m-4}+p^m-p^{m-1}-p^{m-3}+1)$\\
\hline
 $(p-1)(p^{m-1}-p^{(m-1)/2})$ & $\frac{(p^{m+1}+p^{(m+3)/2})(p^{2m}-p^{2m-2}-p^{2m-3}+p^{m-2}+p^{m-3}-1)}{2(p^{2}-1)}$\\
\hline
$(p-1)(p^{m-1}+p^{(m-1)/2})$ & $\frac{(p^{m+1}-p^{(m+3)/2})(p^{2m}-p^{2m-2}-p^{2m-3}+p^{m-2}+p^{m-3}-1)}{2(p^{2}-1)}$\\
\hline
 $(p-1)(p^{m-1}-p^{(m+1)/2})$ & $\frac{(p^{m-3}+p^{(m-3)/2})(p^{m-1}-1)(p^m-1)}{2(p^{2}-1)}$\\
\hline
$(p-1)(p^{m-1}+p^{(m+1)/2})$ & $\frac{(p^{m-3}-p^{(m-3)/2})(p^{m-1}-1)(p^m-1)}{2(p^{2}-1)}$\\
\hline
\end{tabular}
\begin{tablenotes}
\end{tablenotes}
\end{threeparttable}
\end{table*}

\begin{example}
Let $p=3$, $m=5$ and $k=1$. Then the code  $\C_{(p,m,k)}$ is
a  $[242,15,108]$  code over $\gf_3$ with the weight enumerator
\begin{eqnarray*}
1+14520 z^{108}+ 2548260 z^{144}+9740258 z^{162}+ 2038608 z^{180}+7260 z^{216}
\end{eqnarray*}
which confirms the weight distribution in Table \ref{Tb_wd}.
\end{example}

\begin{example}
Let $p=3$, $m=7$ and $k=2$. Then the code  $\C_{(p,m,k)}$ is
a $[2186,21,1296]$  code over $\gf_3$ with the weight enumerator
\begin{eqnarray*}
1+8951670 z^{1296}+ 1732767876 z^{1404}+7102473578 z^{1458}+ 1608998742 z^{1512}+7161336 z^{1620}
\end{eqnarray*}
which confirms the weight distribution in Table \ref{Tb_wd}.
\end{example}

\begin{example}
Let $p=5$, $m=5$ and $k=1$. Then the code  $\C_{(p,m,k)}$ is
a  $[3124,15,2000]$  code over $\gf_3$ with the weight enumerator
\begin{eqnarray*}
1+1218360 z^{2000}+ 3147430000 z^{2400}+24462797524 z^{2500}+ 2905320000 z^{2600}+812240 z^{3000}
\end{eqnarray*}
which confirms the weight distribution in Table \ref{Tb_wd}.
\end{example}

\section{Summary and concluding remarks}\label{sec-conclusion}

In this paper, we studied a family of five-weight cyclic codes. The duals of the cyclic codes have three zeros.
The weight distribution of this family of cyclic codes is completely determined.

Finally we mention that the weight distribution of $\C_{(p,m,k)}$ can also be settled in a more general case
where $m/\gcd(m,k)$ is odd. In what follows we only report the conclusion. The proof is  similar to that of
Theorem \ref{Th_main1}.

\begin{theorem}\label{Th_main2}
Let $\gcd(m,k)=e$, $m/e$ be odd, and $m/e\geq 5$. Let $\C_{(p,m,k)}$ be the code  in (\ref{eqn-def-code-C}). Then $\C_{(p,m,k)}$ is a
cyclic code over $\gf_p$  with parameters
$$
[p^m-1,3m,(p-1)(p^{m-1}-p^{(m+3e-2)/2})].
$$
Furthermore,  the weight distribution of  $\C_{(p,m,k)}$ is given by Table \ref{Tb_wd2}.
\end{theorem}

\begin{table*}[!t]
\renewcommand{\arraystretch}{2}
\centering
\begin{threeparttable}
\caption{Weight Distribution of the Code $\C_{(p,m,k)}$ in Theorem \ref{Th_main2}}\label{Tb_wd2}
\begin{tabular}{|l|l|}
\hline
 Hamming Weight& Frequency\\
\hline
\hline
  $0$ & $1$ \\
  \hline
 $(p-1)p^{m-1}$ & $(p^m-1)(p^{2m}-p^{2m-e}+p^{2m-4e}+p^m-p^{m-e}-p^{m-3e}+1)$\\
\hline
 $(p-1)(p^{m-1}-p^{(m+e-2)/2})$ & $\frac{(p^{m+e}+p^{(m+3e)/2})(p^{2m}-p^{2m-2e}-p^{2m-3e}+p^{m-2e}+p^{m-3e}-1)}{2(p^{2e}-1)}$\\
\hline
$(p-1)(p^{m-1}+p^{(m+e-2)/2})$ & $\frac{(p^{m+e}-p^{(m+3e)/2})(p^{2m}-p^{2m-2e}-p^{2m-3e}+p^{m-2e}+p^{m-3e}-1)}{2(p^{2e}-1)}$\\
\hline
 $(p-1)(p^{m-1}-p^{(m+3e-2)/2})$ & $\frac{(p^{m-3e}+p^{(m-3e)/2})(p^{m-e}-1)(p^m-1)}{2(p^{2e}-1)}$\\
\hline
$(p-1)(p^{m-1}+p^{(m+3e-2)/2})$ & $\frac{(p^{m-3e}-p^{(m-3e)/2})(p^{m-e}-1)(p^m-1)}{2(p^{2e}-1)}$\\
\hline
\end{tabular}
\begin{tablenotes}
\end{tablenotes}
\end{threeparttable}
\end{table*}


\section*{Appendix I}\label{Appendix-I}

{\em Proof of Lemma \ref{lem-mainA}}:

For any $(\bar{a}, \bar{b},\bar{c}) \in \gf_q^3$, let $\bar{N}_{(\bar{a}, \bar{b},\bar{c})}$ denote the number of
solutions $(x,y,u,v) \in \gf_q^4$
of the following system of equations
\begin{eqnarray}\label{eqn-mainlemmaA2}
\left\{
\begin{array}{l}
x+y=\bar{a} \\
x^{d_1}+y^{d_1}=\bar{b} \\
x^{d_2}+y^{d_2}=\bar{c} \\
u+v=-\bar{a} \\
u^{d_1}+v^{d_1}=-\bar{b} \\
u^{d_2}+v^{d_2}=-\bar{c}.
\end{array}
\right.
\end{eqnarray}
It is then obvious that
$$
\mathfrak{N}_4=\sum_{(\bar{a}, \bar{b},\bar{c}) \in \gf_q^3}  \bar{N}_{(\bar{a}, \bar{b},\bar{c})}.
$$

\par
For any $(\bar{a}, \bar{b},\bar{c}) \in \gf_q^3$, let $\hat{N}_{(\bar{a}, \bar{b},\bar{c})}$ denote the number of
solutions $(x,y) \in \gf_q^2$
of the following system of equations
\begin{eqnarray}\label{eqn-mainlemmaA3}
\left\{
\begin{array}{l}
x+y=\bar{a} \\
x^{d_1}+y^{d_1}=\bar{b} \\
x^{d_2}+y^{d_2}=\bar{c}.
\end{array}
\right.
\end{eqnarray}
Since $d_1$ and $d_2$ are odd,  $\bar{N}_{(\bar{a}, \bar{b},\bar{c})}=\left(\hat{N}_{(\bar{a}, \bar{b},\bar{c})}\right)^2$, we have
\begin{eqnarray}\label{eqn-N41}
\mathfrak{N}_4=\sum_{(\bar{a}, \bar{b},\bar{c}) \in \gf_q^3}  \left(\hat{N}_{(\bar{a}, \bar{b},\bar{c})}\right)^2.
\end{eqnarray}
We distinguish among the following three cases to calculate $\hat{N}_{(\bar{a}, \bar{b},\bar{c})}$.

{\textit{Case A}}, when  $\bar{a}=\bar{b}=\bar{c}=0$: In this case,  $\hat{N}_{(0, 0,0)}=q$
since $(x,y)$ is a solution of (\ref{eqn-mainlemmaA3}) if and only if $y=-x$. Thus $(\hat{N}_{(0, 0,0)})^2=q^2$.

{\textit{Case B}}, when  $\bar{a}\neq 0$, and ($\bar{b}=0$ or $\bar{c}=0$): In this case,
it is clear that $\hat{N}_{(\bar{a}, \bar{b},\bar{c})}=0$.

{\textit{Case C}}, when $\bar{a} \ne 0$, $\bar{b} \ne 0$ and $\bar{c} \ne 0$. In this case,
for any given $\bar{a}\neq 0$,  Equation System (\ref{eqn-mainlemmaA3}) has the same number of solutions as
\begin{eqnarray}\label{eqn-mainlemmaB}
\left\{
\begin{array}{l}
x+y =1 \\
x^{d_1} + {y}^{d_1} =b  \\
x^{d_2} + {y}^{d_2} =c
\end{array}
\right.
\end{eqnarray}
where $b={\bar{b}}/{\bar{a}^{d_1}}$ and $c={\bar{c}}/{\bar{a}^{d_2}}$.
Clearly, $(b,c)$ runs over $\gf^*_q\times \gf^*_q$ as $(\bar{b}, \bar{c})$ does.
By Lemma \ref{lem-mainB}, we have
 \begin{eqnarray*}
\sum_{(\bar{a}, \bar{b},\bar{c}) \in (\gf^*_q)^3}\left(\hat{N}_{(\bar{a}, \bar{b},\bar{c})}\right)^2 =(q-1) \left[p^2+(p+1)^2 \frac{q-p}{2(p+1)} + (p-1)^2 \frac{q-p}{2(p-1)} \right]=q(qp-p).
\end{eqnarray*}

Summarizing all the cases above, we have
\begin{eqnarray*}
\mathfrak{N}_4=q^2 + q(qp-p)=q(qp+q-p).
\end{eqnarray*}
This completes the proof.

\mbox

\begin{lemma}\label{lem-mainB}
Let $\mathsf{N}_{(b,c)}$ denote the number of solutions $(x,y) \in \gf_q^2$
of (\ref{eqn-mainlemmaB}),
where $(b, c) \in \gf^*_q\times  \gf^*_q$. Then we have the following conclusions.
\begin{itemize}
\item[B1] $\mathsf{N}_{(1,1)}=p$.
\item[B2] When $(b,c)$ runs over $\gf_q^* \times \gf_q^* \setminus \{(1,1)\}$,
\begin{eqnarray*}
\mathsf{N}_{(b,c)} = \left\{
\begin{array}{ll}
p+1 & \mbox{ for } \frac{q-p}{2(p+1)} \mbox{ times} \\
p-1 & \mbox{ for } \frac{q-p}{2(p-1)} \mbox{ times} \\
0     & \mbox{ for the rest.}
\end{array}
\right.
\end{eqnarray*}
\end{itemize}
\end{lemma}

The proof of Lemma \ref{lem-mainB} is lengthy and technical. We first prove some auxiliary results.

\section*{Auxiliary results for proving Lemma \ref{lem-mainB}}

We prove Lemma \ref{lem-mainB} only for the case that $p \equiv 3 \pmod{4}$. The proof
for the case $p \equiv 1 \pmod{4}$ is similar and omitted. Hence we assume that $p \equiv 3 \pmod{4}$
from now on.

\subsection{Case 1}

In (\ref{eqn-mainlemmaB}) we substitute $(x, y)$ with $(x_1^2, -y_1^2)$ and obtain the
following system of equations
\begin{eqnarray}\label{eqn-mainlemmaB11}
\left\{
\begin{array}{l}
x_1^2-y_1^2=1 \\
x_1^{2d_1}-y_1^{2d_1}=b \\
x_1^{2d_2}-y_1^{2d_2}=c
\end{array}
\right.
\end{eqnarray}
where $b, c \in \gf_q^*$. Our task is to compute the number $\mathtt{N}_{(b,c)}$ of solutions $(x_1, y_1) \in \gf_q^2$
of (\ref{eqn-mainlemmaB11}). To this end, we first compute the number $\mathtt{N}_{b}$ of solutions
$(x_1, y_1) \in \gf_q^2$ of the
following system of equations
\begin{eqnarray}\label{eqn-mainlemmaB12}
\left\{
\begin{array}{l}
x_1^2-y_1^2=1 \\
x_1^{2d_1}-y_1^{2d_1}=b
\end{array}
\right.
\end{eqnarray}
where $b \in \gf_q^*$.

\begin{lemma}\label{lem-case11}
Let symbols and notations be the same as before. As for Equation (\ref{eqn-mainlemmaB12}), we have
\begin{eqnarray*}
\mathtt{N}_{b} = \left\{
\begin{array}{ll}
p-1 & \mbox{ if } b=1  \\
2(p-1) & \mbox{ for } \frac{q-p}{2(p-1)} \mbox{ elements } b \ne 1 \\
0     & \mbox{ for the rest } b \neq 1.
\end{array}
\right.
\end{eqnarray*}
\end{lemma}

\begin{proof}
Let $(x_1,y_1)$ be a solution of the first equation in (\ref{eqn-mainlemmaB12}).
It is clear that $x_1\neq y_1$.
Let $\theta=x_1-y_1$. It then follows that $\theta\in \gf^*_{q}$ and
\begin{eqnarray}\label{theta-case11}
x_1=\frac{\theta+\theta^{-1}}{2}, \ y_1=\frac{\theta^{-1}-\theta}{2}.
\end{eqnarray}
Thus $(x_1,y_1)$ is uniquely determined by $\theta$.
Substituting (\ref{theta-case11}) into the second equation of (\ref{eqn-mainlemmaB12}), we obtain
\begin{eqnarray}\label{eqn-w-1}
\theta^{p^{2k}-1} + \theta^{1-p^{2k}}  =2b.
\end{eqnarray}
Let $w=\theta^{p^{2k}-1}$. Then (\ref{eqn-w-1}) is equivalent to
\begin{equation}\label{eqn-quadratic-w1}
w^2-2{b} w+1=0.
\end{equation}
If (\ref{eqn-quadratic-w1}) has no solution, i.e., $b^2-1$ is not a square in $\gf_q^*$,  then $\mathtt{N}_{b}=0$.
Otherwise, suppose that $w_1$ and $w_2=w_1^{-1}$ are two solutions of (\ref{eqn-quadratic-w1}).
We then have
\begin{eqnarray}\label{eqn-wtheta-1}
\theta^{p^{2k}-1}=w_1
\end{eqnarray}
or
\begin{eqnarray}\label{eqn-wtheta-2}
\theta^{p^{2k}-1}=w_1^{-1}.
\end{eqnarray}
Clearly, (\ref{eqn-wtheta-1}) and (\ref{eqn-wtheta-2}) have the same number of solutions $\theta\in \gf_q$.
Note that $\gcd(p^{2k}-1, q-1)=p-1$. Thus both (\ref{eqn-wtheta-1}) and (\ref{eqn-wtheta-2})
have  no solution or exactly $p-1$ solutions. If $w_1=w_1^{-1}$, then $w_1=\pm 1$ and $b=\pm 1$.
However $-1$ is not a square, thus, $w_1=1$ and $b=1$. In this case, (\ref{eqn-wtheta-1}) and (\ref{eqn-wtheta-2})
become the same equation and have $p-1$ solutions.
If $w_1\neq w_1^{-1}$, then  (\ref{eqn-wtheta-1}) and (\ref{eqn-wtheta-2}) have distinct solutions.

Based on above analysis, we conclude
$$
\mathtt{N}_{1}=p-1 {\textrm{~and~}} \mathtt{N}_{b}=0  {~\textrm{or}~2(p-1)} \textrm{~for~}b\neq 1.
$$
Define
   \[T=\#\{b\in \gf_q: N_{b}=2(p-1)\}.\]
Note that the first equation in (\ref{eqn-mainlemmaB12}) has $q-1$ solutions in $\gf_q$ thanks to
Lemma 6.24 in \cite{Niddle}.
When $(x,y)$
runs through all these solutions, the second equation in (\ref{eqn-mainlemmaB12})  will
give a $2(p-1)$-to-$1$ correspondence
\[(x,y)\mapsto b=x^{p^{2k}+1}-y^{p^{2k}+1}\]
if $N_{b}=2(p-1)$. Therefore
\[(p-1)+2(p-1)T=q-1\]
which leads to
\[T=\frac{q-p}{2(p-1)}.\]
This completes the proof.
\end{proof}

\begin{lemma}\label{lem-case12}
Let symbols and notations be the same as before. As for Equation System (\ref{eqn-mainlemmaB11}), we have
\begin{eqnarray*}
\mathtt{N}_{(b,c)} = \left\{
\begin{array}{ll}
p-1 & \mbox{ if } (b,c)=(1,1)  \\
2(p-1) & \mbox{ for } \frac{q-p}{2(p-1)} \mbox{ pairs } (b,c) \ne (1,1) \\
0     & \mbox{ for the rest pairs} (b,c) \in (\gf_q^* )^2 \setminus \{(1,1)\}.
\end{array}
\right.
\end{eqnarray*}
\end{lemma}

\begin{proof}
Let $(x_1, y_1)$ be any solution of (\ref{eqn-mainlemmaB11}).  Let $\theta=x_1-y_1$. It then follows from
the first equation in (\ref{eqn-mainlemmaB12}) that
\begin{eqnarray}\label{theta-case1}
x_1=\frac{\theta+\theta^{-1}}{2}, \ y_1=\frac{\theta^{-1}-\theta}{2}.
\end{eqnarray}
Using the second and third equations in  (\ref{eqn-mainlemmaB11}), we obtain
\begin{eqnarray*}
\left\{
\begin{array}{l}
b =\frac{1}{2} \left( \theta^{p^{2k}-1} + \theta^{1-p^{2k}}  \right)  \\
c =\frac{1}{2} \left( \theta^{p^{4k}-1} + \theta^{1-p^{4k}}  \right).
\end{array}
\right.
\end{eqnarray*}
Let $w=\theta^{p^{2k}-1}$ and
\begin{eqnarray}\label{eqn-Case1D1}
\left\{
\begin{array}{l}
\tilde{b}= \left( \theta^{p^{2k}-1} + \theta^{1-p^{2k}}  \right)=w+w^{-1} \\
\tilde{c}= \left( \theta^{p^{4k}-1} + \theta^{1-p^{4k}}  \right)=w^{p^{2k}+1}+(w^{-1})^{p^{2k}+1}.
\end{array}
\right.
\end{eqnarray}
Then $w^2-\tilde{b} w +1=0$ and
$$
w=\frac{\tilde{b}}{2} \pm \sqrt{ \left(\frac{\tilde{b}}{2}\right)^2 -1 }.
$$

It follows from the first equation in (\ref{eqn-Case1D1}) that
\begin{eqnarray}\label{eqn-Case1D2}
\tilde{b}^{p^{2k}}=w^{p^{2k}}+(w^{-1})^{p^{2k}}.
\end{eqnarray}

Combining the first equation in (\ref{eqn-Case1D1}) and (\ref{eqn-Case1D2}), we obtain
$$
\tilde{b}^{p^{2k}+1}=\tilde{c} + w^{p^{2k}-1}+(w^{-1})^{p^{2k}-1}.
$$
Whence,
\begin{eqnarray}\label{eqn-426}
\tilde{c} = \tilde{b}^{p^{2k}+1} - \left(  w^{p^{2k}-1}+(w^{-1})^{p^{2k}-1} \right).
\end{eqnarray}
Note that
$$
w=\frac{\tilde{b}}{2} \pm \sqrt{ \left(\frac{\tilde{b}}{2}\right)^2 -1 }
$$
if and only if
$$
w^{-1}=\frac{\tilde{b}}{2} \mp \sqrt{ \left(\frac{\tilde{b}}{2}\right)^2 -1}.
$$
By (\ref{eqn-426}), $\tilde{c}$ is uniquely determined by $\tilde{b}$.
Therefore, $c$ is uniquely determined by $b$.

In addition, it is easily seen that $\tilde{c}=2$ if and only if $\tilde{b}=2$.

Hence the number of solutions of (\ref{eqn-mainlemmaB11}) is the same as that
of (\ref{eqn-mainlemmaB12}). The desired conclusions then follow from Lemma
\ref{lem-case11}.
\end{proof}

\begin{lemma}\label{lem-case13}
Let $\mathtt{M}_{(b,c)}$ denote the number of solutions $(x, y)$ of (\ref{eqn-mainlemmaB})
such that $x$ is a square and $y$ is a nonquare or $y=0$. Then
\begin{eqnarray*}
\mathtt{M}_{(b,c)} = \left\{
\begin{array}{ll}
\frac{p+1}{4}  & \mbox{ if } (b,c)=(1,1)  \\
\frac{p-1}{2} & \mbox{ for } \frac{q-p}{2(p-1)} \mbox{ pairs } (b,c) \ne (1,1) \\
0     & \mbox{ for the rest pairs} (b,c) \in (\gf_q^* )^2 \setminus \{(1,1)\}.
\end{array}
\right.
\end{eqnarray*}
\end{lemma}

\begin{proof}
Consider now the solutions of (\ref{eqn-mainlemmaB11}). If $(x_1, y_1)$ is a solution
of (\ref{eqn-mainlemmaB11}), so are $(-x_1, y_1)$, $(x_1, -y_1)$ and $(-x_1, -y_1)$.
If $y_1 \ne 0$, they are indeed four different solutions of  (\ref{eqn-mainlemmaB11}),
but give only one solution of (\ref{eqn-mainlemmaB}).

Since $-1$ is a quadratic nonresidue in $\gf_q$, $x_1 \ne 0$. However, it is possible that $y_1=0$.
If $y_1=0$, then $(b,c)=(1,1)$. In this case, we have two special solutions $(\pm1, 0)$
of  (\ref{eqn-mainlemmaB11}). They give only one solution of (\ref{eqn-mainlemmaB}).

It then follows from Lemma \ref{lem-case12} that
$$
\mathtt{M}_{(1,1)} = \frac{\mathtt{N}_{(1,1)}-2}{4}+1=\frac{p-3}{4} +1=\frac{p+1}{4}
$$
and
\begin{eqnarray*}
\mathtt{M}_{(b,c)} &=& \frac{\mathtt{N}_{(b,c)}}{4} \\
                            &=&  \left\{
\begin{array}{ll}
\frac{p-1}{2} & \mbox{ for } \frac{q-p}{2(p-1)} \mbox{ pairs } (b,c) \ne (1,1) \\
0     & \mbox{ for the rest pairs} (b,c) \in (\gf_q^* )^2 \setminus \{(1,1)\}.
\end{array}
\right.
\end{eqnarray*}
The proof is then completed.
\end{proof}

\subsection{Case 2}

\begin{lemma}\label{lem-case23}
Let $\mathtt{M}_{(b,c)}$ denote the number of solutions $(x, y)$ of (\ref{eqn-mainlemmaB})
such that $y$ is a square and $x$ is a nonsquare or $x=0$. Then
\begin{eqnarray*}
\mathtt{M}_{(b,c)} = \left\{
\begin{array}{ll}
\frac{p+1}{4}  & \mbox{ if } (b,c)=(1,1)  \\
\frac{p-1}{2} & \mbox{ for } \frac{q-p}{2(p-1)} \mbox{ pairs } (b,c) \ne (1,1) \\
0     & \mbox{ for the rest pairs} (b,c) \in (\gf_q^* )^2 \setminus \{(1,1)\}.
\end{array}
\right.
\end{eqnarray*}
\end{lemma}

This case is symmetric to Case 1. Hence the proof of this lemma is similar to that of Lemma \ref{lem-case13}
and is omitted.

\subsection{Case 3}

In (\ref{eqn-mainlemmaB}) we substitute $(x, y)$ with $(x_1^2, y_1^2)$ and obtain the
following system of equations
\begin{eqnarray}\label{eqn-mainlemmaB31}
\left\{
\begin{array}{l}
x_1^2+y_1^2=1 \\
x_1^{2d_1}+y_1^{2d_1}=b \\
x_1^{2d_2}+y_1^{2d_2}=c
\end{array}
\right.
\end{eqnarray}
where $b, c \in \gf_q^*$. Our task is to compute the number $N_{(b,c)}$ of solutions $(x_1, y_1) \in \gf_q^2$
of (\ref{eqn-mainlemmaB31}). To this end, we first compute the number $\mathtt{N}_{b}$ of solutions
$(x_1, y_1) \in \gf_q^2$ of the
following system of equations
\begin{equation}\label{origin}
\left\{
\begin{array}{l}
x^2+y^2=1 \\
x^{p^{2k}+1}+y^{p^{2k}+1}=b
\end{array}
\right.
\end{equation}
where $b \in \gf_q^*$.

\begin{lemma}\label{lem-case31}
Let symbols and notations be the same as before. As for Equation (\ref{origin}), we have
\[
N_{b}= \left\{
\begin{array}{ll}
p+1  &\text{if}\; b=1 \\
2(p+1)  & \text{for}\;\frac{q-p}{2(p+1)}\;\text{elements}\; b\neq 1 \\
0 &\text{for the rest}\; b\neq 1.
\end{array}
\right.
\]
\end{lemma}

\begin{proof}
Choose $t\in \gf_{p^2}$ such that $t^2=-1$. From
\begin{equation}\label{first}
x^2+y^2=1
\end{equation}
we can assume
\begin{equation}\label{theta}
x=\frac{\theta+\theta^{-1}}{2},\qquad y=\frac{t(\theta-\theta^{-1})}{2}
\end{equation}
with $\theta\in \gf_{q^2}^*$. It is easy to see that all the
solutions $(x,y)\in \gf^2_q$ of (\ref{first}) can be expressed as in
(\ref{theta}) with a unique $\theta\in \gf_{q^2}^*$.
Substituting (\ref{theta}) into
\begin{equation}\label{second}
x^{p^{2k}+1}+y^{p^{2k}+1}=b,
\end{equation}
we obtain
\begin{equation}\label{b}
\theta^{p^{2k}-1}+\theta^{1-p^{2k}}=2b.
\end{equation}
Denote by $w=\theta^{p^{2k}-1}$. Then (\ref{b}) is equivalent to
\begin{equation}\label{quadratic}
w^2-2b w+1=0.
\end{equation}
Let $w_1$ and $w_2=w_1^{-1}$ be two solutions of (\ref{quadratic}).
Then we have $w_1\in \gf_{q^2}^*$.

From (\ref{theta}) and $x\in \gf_q$ we have
$$
\theta+\theta^{-1}=(\theta+\theta^{-1})^q=\theta^q+\theta^{-q}
$$
which implies
\[\theta^{q+1}=1\;\text{or}\;\theta^{q-1}=1.\]

\begin{itemize}
     \item If $\theta^{q+1}=1$, then $y^q=\frac{t^q(\theta^q-\theta^{-q})}{2}=\frac{t(\theta-\theta^{-1})}{2}=y$ since $t^q=-t$. It follows that $y\in \gf_q$.
     For a fixed $b$, recall that $w_1$ and $w_2=w_1^{-1}$ are two solutions of
     (\ref{quadratic}). Then we have
     \begin{equation}\label{end1}
      \theta^{p^{2k}-1}=w_1,\quad \theta^{q+1}=1
     \end{equation}
or
     \begin{equation}\label{end1'}
      \theta^{p^{2k}-1}=w_1^{-1},\quad \theta^{q+1}=1.
     \end{equation}
     If $\theta_1$ and $\theta_2$ are two solutions of (\ref{end1}),
     then $(\theta_1/\theta_2)^{p^{2k}-1}=(\theta_1/\theta_2)^{q+1}=1$
     which is equivalent to $(\theta_1/\theta_2)^{p+1}=1$. As a
     consequence,
     if (\ref{end1}) has solutions, then it has exactly $p+1$
     solutions.

     If $w_1=w_1^{-1}$, then (\ref{end1'}) is the same with
     (\ref{end1}) and apparently it gives no more solutions. In this
     case $w_1=\pm 1$ and $b=\pm 1$. But $b=-1$ can be excluded
     since, otherwise, $w_1=-1$, then
     $\theta^{p^{2k}-1}=-1$ which contradicts to $\theta\in \gf_{p^{2m}}$.
     The remaining case is $b=1$ which corresponds to $w_1=1$. In
     this case we have $p+1$ solutions of $\theta$ which gives
     exactly the same number of solutions of (\ref{origin}).

     If $w_1\neq w_1^{-1}$, then (\ref{end1'}) has the same number
     of solutions as (\ref{end1}) and moreover, their solutions are
     distinct. Therefore (\ref{end1}) and (\ref{end1'}) both have
     $p+1$ solutions or no solutions in $\gf_{q^2}$.
     \item If $\theta^{q-1}=1$ and $\theta^{q+1}\neq 1$, then $\theta\in \gf_q^*$. Note that $t\notin\gf_q^*$,
  $y=\frac{t(\theta-\theta^{-1})}{2}$ is not in $\gf_q^*$ except for
  $\theta=\theta^{-1}=\pm 1$. But the exception case will not occur
  since $\theta^{q+1}\neq 1$.
   \end{itemize}

Summarizing up, we conclude
  \[N_1=p+1 \textrm{~and~} N_{b}=0\;\text{or}\;2(p+1)\;\text{for}\; b\neq 1.\]
Define
   \[T=\#\{b\in \gf_q: N_{b}=2(p+1)\}.\]
Note that (\ref{first}) has $q+1$ solutions in $\gf_q$ thanks to
Lemma 6.24 in \cite{Niddle}. When $(x,y)$
runs through all these solutions, the equation (\ref{second}) will
give a $2(p+1)$-to-$1$ correspondence
\[(x,y)\mapsto b=x^{p^{2k}+1}+y^{p^{2k}+1}\]
if $N_{b}=2(p+1)$. Therefore
\[(p+1)+2(p+1)T=q+1\]
which implies
\[T=\frac{q-p}{2(p+1)}.\]
The proof is now finished.
\end{proof}

\begin{lemma}\label{lem-case32}
Let symbols and notations be the same as before. As for Equation System (\ref{eqn-mainlemmaB31}), we have
\[
N_{(b, c)}= \left\{
\begin{array}{ll}
p+1  &\text{if } (b, c)=(1,1)  \\
2(p+1) & \text{for~} \frac{q-p}{2(p+1)}\;\text{ pairs } (b, c) \neq (1,1) \\
0  &\text{for the rest } (b, c) \neq (1,1).
\end{array}
\right.
\]
\end{lemma}

\begin{proof}
The proof of this lemma is similar to that of Lemma \ref{lem-case12} and is derived from Lemma
\ref{lem-case31}. The details of the proof is omitted here.
\end{proof}

\begin{lemma}\label{lem-case33}
Let $M_{(b,c)}$ denote the number of solutions $(x, y)$ of (\ref{eqn-mainlemmaB})
such that both $x$ and $y$ are squares. Then
\begin{eqnarray*}
M_{(b,c)} = \left\{
\begin{array}{ll}
\frac{p+5}{4}  & \mbox{ if } (b,c)=(1,1)  \\
\frac{p+1}{2} & \mbox{ for } \frac{q-p}{2(p+1)} \mbox{ pairs } (b,c) \ne (1,1) \\
0     & \mbox{ for the rest pairs} (b,c) \in (\gf_q^* )^2 \setminus \{(1,1)\}.
\end{array}
\right.
\end{eqnarray*}
\end{lemma}

\begin{proof}
Consider now the solutions of (\ref{eqn-mainlemmaB31}). If $(x_1, y_1)$ is a solution
of (\ref{eqn-mainlemmaB31}), so are $(-x_1, y_1)$, $(x_1, -y_1)$ and $(-x_1, -y_1)$.
If $x_1y_1 \ne 0$, they are indeed four different solutions of  (\ref{eqn-mainlemmaB31}),
but give only one solution of (\ref{eqn-mainlemmaB}).

However, it is possible that $x_1y_1=0$.
If $(b,c)=(1,1)$, Equation (\ref{eqn-mainlemmaB31}) has four special solutions $(\pm 1, 0)$
and $(0, \pm 1)$. They give only two solutions of (\ref{eqn-mainlemmaB}).
It then follows from Lemma \ref{lem-case32} that
$$
M_{(1,1)} = \frac{N_{(1,1)}-4}{4}+2=\frac{p+5}{4}.
$$

If $(b,c) \ne (1,1)$, then the four distinct solutions $(\pm x_1, \pm y_1)$ give only one
solution of (\ref{eqn-mainlemmaB}). In this case, it then follows from Lemma \ref{lem-case32} that
\begin{eqnarray*}
M_{(b,c)} &=& \frac{N_{(b,c)}}{4} \\
                            &=&  \left\{
\begin{array}{ll}
\frac{p+1}{2} & \mbox{ for } \frac{q-p}{2(p+1)} \mbox{ pairs } (b,c) \ne (1,1) \\
0     & \mbox{ for the rest pairs} (b,c) \in (\gf_q^* )^2 \setminus \{(1,1)\}.
\end{array}
\right.
\end{eqnarray*}
The proof is then completed.
\end{proof}

\subsection{Case 4}

\begin{lemma}\label{lem-case43}
Let $M_{(b,c)}$ denote the number of solutions $(x, y)$ of (\ref{eqn-mainlemmaB})
such that both $x$ and $y$ are either nonsquares or zero. Then
\begin{eqnarray*}
M_{(b,c)} = \left\{
\begin{array}{ll}
\frac{p-3}{4}  & \mbox{ if } (b,c)=(1,1)  \\
\frac{p+1}{2} & \mbox{ for } \frac{q-p}{2(p+1)} \mbox{ pairs } (b,c) \ne (1,1) \\
0     & \mbox{ for the rest pairs} (b,c) \in (\gf_q^* )^2 \setminus \{(1,1)\}.
\end{array}
\right.
\end{eqnarray*}
\end{lemma}

\begin{proof}
The proof of this lemma is similar to that of Lemma \ref{lem-case33} and is omitted here.
\end{proof}

\section*{The proof of Lemma \ref{lem-mainB}}

Note that the solutions $(1, 0)$ and $(0,1)$ of (\ref{eqn-mainlemmaB}) are counted more
than once in Cases 1, 2, 3 and 4. By analyzing the proofs of Lemmas \ref{lem-case13},
\ref{lem-case23}, \ref{lem-case33} and \ref{lem-case43}, we have
$$
\mathtt{N_{(1,1)}}=\frac{p-3}{4} + \frac{p-3}{4} + \frac{p+1}{4} + \frac{p-3}{4} + 2 =p.
$$

When $(b, c) \ne (1,1)$, $\mathtt{N_{(b,c)}}$ is the sum of the solutions given in Lemmas
\ref{lem-case13}, \ref{lem-case23}, \ref{lem-case33} and \ref{lem-case43}. This completes
the proof.


\begin{thebibliography}{99}


\bibitem{Delsarte} P. Delsarte, ``On subfield subcodes of modified Reed-Solomon codes,''
{\em IEEE Trans. Inform. Theory,} vol. IT-21, no. 5, pp. 575--576, Sep. 1975.

\bibitem{Ding12} C. Ding and J. Yang, ``Hamming weights in irreducible cyclic codes,''
\emph {Discrete Mathematics,} vol. 313, no. 4, pp. 434--446, Feb. 2013.

\bibitem{Ding11} C. Ding, Y. Liu, C. Ma, and L. Zeng, ``The weight distributions of the duals of cyclic codes with
two zeros," {\em IEEE Trans. Inform. Theory,} vol. 57, no. 12, pp. 8000--8006, Dec. 2011.


\bibitem{Feng07} K. Feng and J. Luo, ``Weight distribution of some reducible cyclic codes,''
{\em Finite Fields Appl.,} vol. 14, no. 4, pp. 390--409,  Apr. 2008.

\bibitem{Feng12} T. Feng, ``On cyclic codes of length $2^{2^r}-1$ with two zeros whose dual codes have three weights,"
{\em Des. Codes Cryptogr.}, vol. 62,  pp. 253--258, 2012.

\bibitem{Klapperodd} A. Klapper, ``Cross-correlations of quadratic form sequences in odd characteristic,''
{\em Des. Codes Cryptogr.,}, vol. 3, no. 4, pp. 289--305, June 1997.

\bibitem{Klov} T. Kl{\o}ve, {\em Codes for Error Detection,}  World Scientific, 2007.

\bibitem{Li12} S. X. Li, S. H. Hu, T. Feng, and G. Ge, ``The weight distribution of a class
of cyclic codes related to Hermitian for Graphs," arXiv:1212.6371, 2012.

\bibitem{Niddle} R. Lidl and  H. Niederreiter, {\em Finite Fields,} Encyclopedia of Mathematics, Vol. 20,
Cambridge University Press, Cambridge, 1983.

\bibitem{Luo081} J. Luo and K. Feng, ``On the weight distribution of two classes of cyclic
codes,'' {\em IEEE Trans. Inform. Theory,}  vol. 54, no. 12, pp. 5332--5344, Dec. 2008.

\bibitem{Luo082} J. Luo and K. Feng, ``Cyclic codes and sequences from generalized Coulter-Matthews
function,'' {\em IEEE Trans. Inform. Theory,}  vol. 54, no. 12, pp. 5345--5353, Dec. 2008.

\bibitem{Ma11} C. Ma, L. Zeng, Y. Liu, D. Feng, and C. Ding, ``The weight enumerator of a class of cyclic
codes,'' {\em IEEE Trans. Inform. Theory,} vol. 57, no. 1, pp. 397--402, Jan. 2011.


\bibitem{Tang} X. H. Tang, P. Udaya, P. Z. Fan, ``A new family of nonbinary sequences with
three-level correlation property and large linear span,''  {\em IEEE
Trans. Inform. Theory,}  {vol. 51}, pp. 2906--2914,  Aug. 2005.

\bibitem{Trachtenberg} H. M. Trachtenberg, {\em On the crosscorrelation functions of maximal
linear recurring sequences,} Ph.D. dissertation, Univ. South. Calif., Los Angels, 1970.

\bibitem{Wang12}B. Wang, C. Tang, Y. Qi, Y. X. Yang, and M. Xu,  ``The weight distributions of cyclic codes and elliptic curves,''
\emph{IEEE Trans. Inform. Theory,} {vol. 58}, no. 12, pp. 7253--7259, Dec. 2012.

\bibitem{Xiong1} M. Xiong, ``The weight distributions of a class of cyclic codes,"
{\em Finite Fields Appl.,} vol. 18, no. 5, pp. 933--945, Sep. 2012.

\bibitem{YCD} J. Yuan, C. Carlet and C. Ding, ``The weight distribution of a class of linear codes from perfect
nonlinear functions,'' \emph{IEEE Trans. Inform. Theory,} {vol. 52}, no. 2, pp. 712--717, Feb. 2006.

\bibitem{Zeng10} X. Zeng, L. Hu, W. Jiang, Q. Yue and X. Cao, ``Weight distribution of a $p$-ary cyclic code,''
{\em Finite Fields Appl.,}  vol. 16, no. 1, pp. 56--73,  Jan. 2010.

\end{thebibliography}
\end{document}